\documentclass{article}
\usepackage{amsmath}
\usepackage{amsthm}
\usepackage{amssymb}
\usepackage{graphicx}
\usepackage[utf8]{inputenc}
\usepackage[onehalfspacing]{setspace}

\newcounter{saveeqns}
\newcommand{\alpheqns}{\setcounter{saveeqns}{\value{equation}}%
\setcounter{equation}{0}%
\renewcommand{\theequation}{\mbox{S\arabic{equation}}}}
\newcommand{\reseteqns}{\setcounter{equation}{\value{saveeqns}}%
\renewcommand{\theequation}{\arabic{equation}}}

\newcounter{saveeqngs}
\newcommand{\alpheqngs}{\setcounter{saveeqngs}{\value{equation}}%
\setcounter{equation}{3}%
\renewcommand{\theequation}{\mbox{GS\arabic{equation}}}}
\newcommand{\reseteqngs}{\setcounter{equation}{\value{saveeqngs}}%
\renewcommand{\theequation}{\arabic{equation}}}

\newcounter{saveeqngss}
\newcommand{\alpheqngss}{\setcounter{saveeqngss}{\value{equation}}%
\setcounter{equation}{3}%
\renewcommand{\theequation}{\mbox{GS\arabic{equation}'}}}
\newcommand{\reseteqngss}{\setcounter{equation}{\value{saveeqngss}}%
\renewcommand{\theequation}{\arabic{equation}}}

 \theoremstyle{plain}
\newtheorem{theo}{Theorem}

\newtheorem{pro}[theo]{Proposition}
\newtheorem{lemma}[theo]{Lemma}
\newtheorem*{problem}{Problem}

\title{Tsallis entropy and generalized Shannon additivity}
\author{Sonja J\"ackle and Karsten Keller}
\date{\today}
\begin{document}
\maketitle

\begin{abstract}
The Tsallis entropy given for a positive parameter $\alpha$ can be considered as a modification of the classical Shannon entropy. For the latter, corresponding to $\alpha=1$, there exist many axiomatic characterizations. One of them based on the well-known Khinchin-Shannon axioms has been simplified several times and adapted to Tsallis entropy, where the axiom of (generalized) Shannon additivity is playing a central role. The main aim of this paper is to discuss this axiom in the context of Tsallis entropy. We show that it is sufficient for characterizing Tsallis entropy with the exceptions of cases $\alpha=1,2$ discussed separately.
\end{abstract}

\section{Introduction}
\paragraph{Some history.} In 1988 Tsallis \cite{Tsallis1988} generalized the Boltzmann-Gibbs entropy 
\begin{align*}
S=-k_B\sum_{i=1}^{n}p_i\ln p_i,
\end{align*}
describing classical thermodynamical ensembles with microstates of probabilities $p_i$, by the entropy
\begin{align*}
S^\alpha=k_B\frac{1-\sum_{i=1}^{n}p_i^\alpha}{\alpha -1}
\end{align*}
for $0<\alpha\neq 1$ in the sense that $\lim_{\alpha\to 1}S^\alpha=S$. 
Here $k_B$ is the Boltzmann constant (which as only a multiplicative constant will not be considered in the following).
Many physicists argue that this was a breakthrough in thermodynamics
since the extension allows better describing systems out of equilibrium and systems with strong correlations between microstates, but there is also critizism on the application of Tsallis' concept (compare \cite{Cartwright2014, Tsallis2016}). In information theory pioneered by Shannon, the Boltzmann-Gibbs entropy is one of the central concept. We follow the usual practice to call it Shannon entropy. Also note that Tsallis' entropy concept coincides up to a constant with the Havrda-Charv\'{a}t entropy \cite{HavrdaCharvat1967} given in 1967 in an information theoretical context.

There have been given many axiomatic characterications of Tsallis' entropy starting from such of the classical Shannon entropy (see below). One important axiom called (generalized) Shannon additivity is extensively discussed and shown to be sufficient in some sense in this paper. 

\paragraph{Tsallis entropy.} In the following, let
$\triangle_n=\{(p_1,p_2,\ldots ,p_n)\in ({\mathbb R}^+)^n, \sum_{i=1}^n p_i=1\}$ for $n\in {\mathbb N}$ be the set of all $n$-dimensional stochastic vectors and
$\triangle=\bigcup_{n\in {\mathbb N}}\triangle_n$ be the set of all stochastic vectors, where ${\mathbb N}=\{1,2,3,\ldots\}$ and ${\mathbb R}^+$
are the sets of natural numbers and of nonnegative real numbers, respectively.
Given $\alpha>0$ with $\alpha\neq 1$, the Tsallis entropy of a stochastic vector $(p_1,p_2,\ldots ,p_n)$ of some dimension $n$ is defined by
\begin{align*}
H(p_1,p_2,\ldots,p_n) = \frac{1-\sum_{i=1}^{n}p_i^\alpha}{\alpha -1}.
\end{align*}
In the case $\alpha=1$, the value $H\left(p_1,\ldots,p_n\right)$ is not defined, but the limit of it as $\alpha$ approaches to $1$ is
\begin{align*}
H(p_1,p_2,\ldots,p_n) = -\sum_{i=1}^{n}p_i\ln p_i,
\end{align*}
which provides the classical Shannon entropy. In so far Tsallis entropy can be considered as a genera\-lization of the Shannon entropy and so it is not surprising that various
axiomatic characterizations of the latter one have been tried to generalize to the Tsallis entropy. 

\paragraph{Axiomatic characterizations.}
One line of characterizations mainly followed by Suyari \cite{Suyari2004} and discussed in this paper has its origin in the {\it Shannon-Khinchin axioms} of Shannon entropy (see \cite{Shannon1948} and \cite{Khinchin1957}). Note that other characterizations of Tsallis entropy are due to
dos Santos \cite{dosSantos1997}, Abe \cite{Abe2004} and Furuichi \cite{Furuichi2005}. For some general discussion of axiomatization of entropies see \cite{Csiszar2008}.

A map $H: \triangle\rightarrow {\mathbb R}^+$ is the Shannon entropy up to a multiplicative positive constant if
it satisfies the following axioms:
\alpheqns
\begin{eqnarray}
H\mbox{ is continuous on }\triangle_n\mbox{ for all }n\in {\mathbb N},\label{ks1}\\
H\left (\frac{1}{n},\frac{1}{n},\ldots ,\frac{1}{n}\right )\geq H(p_1,p_2,\ldots,p_n)\mbox{ for all }(p_1,p_2,\ldots,p_n)\in\triangle_n,\label{ks2}\\
H(p_1,p_2,\ldots,p_n,0)=H(p_1,p_2,\ldots,p_n)\mbox{ for all }(p_1,p_2,\ldots,p_n)\in\triangle,\label{ks3}\\
H(p_{1,1},...,p_{1,m_1},p_{2,1},...,p_{2,m_2},...,p_{n,1},...,p_{n,m_n})= H(p_1,...,p_n)+\sum\limits_{i=1}^{n}p_i H\left(\frac{p_{i,1}}{p_i},...,\frac{p_{i,m_i}}{p_i}\right)\nonumber\\
\mbox{for all }(p_{1,1},...,p_{1,m_1},p_{2,1},...,p_{2,m_2},...,p_{n,1},...,p_{n,m_n})\in \triangle;\ n,m_1,\ldots ,m_n\in {\mathbb N}\nonumber\label{ks4}\\
\mbox{and }p_i=\sum\limits_{j=1}^{m_i}p_{i,j};\ j=1,...,m_i.
\end{eqnarray}
\reseteqns
Axiom \eqref{ks4} called {\em Shannon additivity} is playing a key role in the characterization of the Shannon entropy and an interesting result given by Suyari \cite{Suyari2004} says that its generalization
\alpheqngs
\begin{eqnarray}\label{ks4v}
H(p_{1,1},...,p_{1,m_1},p_{2,1},...,p_{2,m_2},...,p_{n,1},...,p_{n,m_n})= H(p_1,...,p_n)+\sum\limits_{i=1}^{n}p_i^\alpha H\left(\frac{p_{i,1}}{p_i},...,\frac{p_{i,m_i}}{p_i}\right)\nonumber\\
\mbox{for all }(p_{1,1},...,p_{1,m_1},p_{2,1},...,p_{2,m_2},...,p_{n,1},...,p_{n,m_n})\in \triangle;\ n,m_1,\ldots ,m_n\in {\mathbb N}\hspace*{3mm}\nonumber\\
\mbox{and }p_i=\sum\limits_{j=1}^{m_i}p_{i,j};\ j=1,...,m_i\hspace*{3mm}
\end{eqnarray}
\reseteqngs
for $\alpha\neq 1$ provides the Tsallis entropy for this $\alpha$.

More precisely, if $H: \triangle\rightarrow {\mathbb R}^+$ satisfies \eqref{ks1}, \eqref{ks2}, \eqref{ks3} and \eqref{ks4v}, then $c(\alpha)H$ is the Tsallis entropy for some positive constant $c(\alpha)$. The full result of Suyari, which was slightly corrected by Ili\'{c} et al. \cite{IlicEtAl2013} includes a characterization of the map $\alpha\rightarrow c(\alpha)$ under the assumption that $H$ also depends continuously on $\alpha\in {\mathbb R}^+\setminus\{0\}$. We do not discuss this characterization, but we note here that the results below also provide an immediate simplification of the whole result of Suyari and Ili\'{c} et al.

\paragraph{The main result.} In this paper, we study the role of generalized Shannon additivity in characterizing Tsallis entropy, where for $\alpha\in {\mathbb R}^+\setminus\{0\}$ and $H: \triangle\rightarrow {\mathbb R}$ we also consider the slightly relaxed property that
\alpheqngss
\begin{align}\label{qadd}
H(p_1,...,p_{j-1},p_j,p_{j+1},p_{j+2},\ldots ,p_n)=H(p_1,...,p_{j-1},p_j+p_{j+1},p_{j+2},\ldots ,p_n)\nonumber\\
+\,(p_j+p_{j+1})^\alpha H\left(\frac{p_j}{p_j+p_{j+1}},\frac{p_{j+1}}{p_j+p_{j+1}}\right)\hspace{-0.5cm}\nonumber\\
\mbox{for all }(p_1,p_2,\ldots, p_n\in\triangle);\, n\in {\mathbb N}
;\, j=1,2,\ldots, n-1.
\hspace*{-0.25cm}
\end{align}
\reseteqngss
It turns out that this property basically is enough for characterizing the Tsallis entropy for $\alpha\in {\mathbb R}^+\setminus\{0,1,2\}$ and with a further weak assumption in the cases $\alpha=1,2$. As already mentioned, the statement (iii) for $\alpha=1$ is an immediate consequence of a characterization of Shannon entropy by Diderrich \cite{Diderrich1975} simplifying an axiomatization given by Faddeev \cite{Faddeev1956} (see below).
\begin{theo}\label{mainth}
Let $H: \triangle\rightarrow {\mathbb R}$ be given with \eqref{ks4v}, or a bit weaker \eqref{qadd}, for $\alpha\in {\mathbb R}^+\setminus\{0\}$. Then the following holds:
\begin{enumerate}
\item[(i)] If $\alpha\neq 1,2$, then
\begin{align}\label{maini}
H(p_1,p_2,\ldots ,p_n)=H\left(\frac{1}{2},\frac{1}{2}\right)\frac{1-\sum_{i=1}^n p_i^{\alpha}}{1-2^{1-\alpha}}\mbox{ for all }(p_1,p_2,\ldots ,p_n)\in\triangle.
\end{align}
\item[(ii)] If $\alpha=2$, then the following statements are equivalent:
\begin{enumerate}
\item It holds
\begin{align*}
H(p_1,p_2,\ldots ,p_n)=2H\left(\frac{1}{2},\frac{1}{2}\right)\left(1-\sum_{i=1}^n p_i^2\right)\mbox{ for all }(p_1,p_2,\ldots ,p_n)\in\triangle,
\end{align*}
\item $H$ is bounded on $\triangle_2$,
\item $H$ is continuous on $\triangle_2$,
\item $H$ is symmetric on $\triangle_2$,
\item $H$ does not change the signum on $\triangle_2$.
\end{enumerate}
\item[(iii)] If $\alpha=1$, then the following statements are equivalent:
\begin{enumerate}
\item It holds
\begin{align*}
H(p_1,p_2,\ldots ,p_n)=-H\left(\frac{1}{2},\frac{1}{2}\right)\left(\sum_{i=1}^n p_i\log_2 p_i\right)\mbox{ for all }(p_1,p_2,\ldots ,p_n)\in\triangle,
\end{align*}
\item
$H$ is bounded on $\triangle_2$.
\end{enumerate}
\end{enumerate}
\end{theo}
Note that statement (iii) is given here only for reasons of completeness. It follows from a result of Diderrich \cite{Diderrich1975}.

The paper is organized as follows. Section 2 is devoted to the proof of the main result. It will turn out that most of the substantial work is related to stochastic vectors contained in $\triangle_2\cup\triangle_3$ and that the generalized Shannon additivity performs as a bridge to stochastic vectors longer than $2$ or $3$. Section 3 completes the discussion. In particular, the Tsallis entropy for $\alpha=1,2$ on rational vectors is discussed and an open problem is formulated.

\section{Proof of the main result}
We start with investigating the relationship of $H(p_1,p_2)$ and $H(p_2,p_1)$ for $(p_1,p_2)\in\triangle_2$.
\begin{lemma}\label{prelem}
Let $\alpha\in {\mathbb R}^+\setminus\{0\}$ and $H: \triangle\rightarrow {\mathbb R}$ satisfy \eqref{qadd}. Then for all
$(p_1,p_2)\in\triangle_2$ it follows
\begin{align}\label{prefo}
(1-3\cdot 2^{-\alpha})H(p_1,p_2)+2^{-\alpha}H(p_2,p_1)
=H\left(\frac{1}{2},\frac{1}{2}\right)(1-p_1^{\alpha}-p_2^{\alpha}),
\end{align}
in particular for $\alpha = 1$
\begin{align}\label{prefo2}
H(p_1,p_2)=H(p_2,p_1)
\end{align}
and for $\alpha = 2$
\begin{align}\label{prefo3}
H(p_1,p_2)+H(p_2,p_1)
=4H\left(\frac{1}{2},\frac{1}{2}\right)(1-p_1^2-p_2^2).
\end{align}
Moreover it holds
\begin{align}\label{prefo4}
H(1)=0.
\end{align}
\end{lemma}
\begin{proof}
First of all, note that \eqref{prefo4} is an immediate consequence of \eqref{qadd} implying
\begin{align*}
H(1,0)=H(1)+1^\alpha H(1,0).\hspace*{4cm}
\end{align*}
Further, two different applications of \eqref{qadd} to $H\left (\frac{1}{2},\frac{1}{2},0\right)$ provide
\begin{align*}
H\left (\frac{1}{2},\frac{1}{2}\right)+\left (\frac{1}{2}\right )^{\hspace{-1mm}\alpha}H(1,0)=H\left (\frac{1}{2},\frac{1}{2},0\right)=H(1,0)+1^\alpha H\left(\frac{1}{2},\frac{1}{2}\right).
\end{align*}
Therefore $H(1,0)=0$, and since similarly one gets $H(0,1)=0$, in the following we can assume that $p_1,p_2\neq 0$.

Applying \eqref{qadd} three times, one obtains
\begin{align}\label{f1}
H(p_1,p_2)+(p_1^{\alpha}+p_2^{\alpha})H\left(\frac{1}{2},\frac{1}{2}\right)=H\left(\frac{p_1}{2},\frac{p_1}{2},\frac{p_2}{2},\frac{p_2}{2}\right)
=H\left(\frac{p_1}{2},\frac{1}{2},\frac{p_2}{2}\right)+2^{-\alpha}H(p_1,p_2)
\end{align}
and in the same way
\begin{align}\label{f2}
H\left(\frac{p_1}{2},\frac{1}{2},\frac{p_2}{2}\right)+2^{-\alpha}H(p_2,p_1)=H\left(\frac{p_1}{2},\frac{p_2}{2},\frac{p_1}{2},\frac{p_2}{2}\right)
=H\left(\frac{1}{2},\frac{1}{2}\right)+2^{1-\alpha}H(p_1,p_2).
\end{align}
Transforming \eqref{f2} to the term $H\left(\frac{p_1}{2},\frac{1}{2},\frac{p_2}{2}\right)$ and then substituting this term in \eqref{f1}, provides
\begin{align*}
H(p_1,p_2)+(p_1^{\alpha}+p_2^{\alpha})H\left(\frac{1}{2},\frac{1}{2}\right )
=H\left (\frac{1}{2},\frac{1}{2}\right )+3\cdot 2^{-\alpha}H(p_1,p_2)-2^{-\alpha}H(p_2,p_1),
\end{align*}
which is equal to \eqref{prefo}. Statements \eqref{prefo2} and \eqref{prefo3} follow immediately from equation \eqref{prefo}.
\end{proof}
In the case $\alpha=1$ condition \eqref{qadd} implies that the order of components of a stochastic vector does not make a difference for $H$:
\begin{lemma}\label{permut}
Let $H: \triangle\rightarrow {\mathbb R}$ satisfy \eqref{qadd} for $\alpha=1$. Then $H$ is permutation-invariant, meaning that $H(p_1,p_2\ldots ,p_n)=H(p_{\pi(1)},p_{\pi(2)}\ldots ,p_{\pi(n)})$ for each $(p_1,p_2,\ldots ,p_n)\in\triangle;\, n\in {\mathbb N}$ and each permutation $\pi$ of $\{1,2,\ldots ,n\}$.
\end{lemma}
\begin{proof}
It suffices to show that
\begin{align*}
H(p_1,...,p_{j-1},p_j,p_{j+1},p_{j+2},\ldots ,p_n)=H(p_1,...,p_{j-1},p_{j+1},p_j,p_{j+2},\ldots ,p_n)\\
\mbox{for all }(p_1,p_2,\ldots, p_n)\in\triangle;\, n\in {\mathbb N}
;\, j=1,2,\ldots, n-1
\end{align*}
For $n<3$ this has been shown in Lemma \ref{prelem} (see \eqref{prefo2}), for $n\geq 3$ it follows directly from \eqref{qadd} and from Lemma \ref{prelem}.
\end{proof}
The following lemma provides the substantial part of the proof of Theorem \ref{mainth}.
\begin{lemma}\label{beginlem}
For $H: \triangle\rightarrow {\mathbb R}$ satisfying \eqref{qadd} with $\alpha\in {\mathbb R}^+\setminus\{0,1\}$, the following holds:
\begin{enumerate}
\item[(i)] If $\alpha\neq 2$, then
\begin{equation*}
H(p_1,p_2)=H\left(\frac{1}{2},\frac{1}{2}\right)\frac{1-p_1^{\alpha}-p_2^{\alpha}}{1-2^{1-\alpha}}\mbox{ for all }(p_1,p_2)\in \triangle_2.
\end{equation*}
\item[(ii)] If $\alpha=2$, then the following statements are equivalent:
\begin{enumerate}
\item[(a)] It holds
\begin{align*}
H(p_1,p_2)
=2H\left(\frac{1}{2},\frac{1}{2}\right)(1-p_1^2-p_2^2)\mbox{ for all }(p_1,p_2)\in \triangle_2,
\end{align*}
\item[(b)] $H$ is symmetric on $\triangle_2$, meaning that $H(p_1,p_2)=H(p_2,p_1)$ for all $(p_1,p_2)\in \triangle_2$,
\item[(c)] $H$ is continuous on $\triangle_2$,
\item[(d)] $H$ is bounded on $\triangle_2$,
\item[(e)] $H$ is nonnegative or nonpositive on $\triangle_2$.
\end{enumerate}
\end{enumerate}
\end{lemma}
\begin{proof}
We first show (i). Let $\alpha\neq 2$ and $(p_1,p_2)\in\triangle_2$.
Changing the role of $p_1$ and $p_2$ in \eqref{prefo}, by Lemma \ref{prelem}
one obtains
\begin{align}\label{neufo1}
(1-3\cdot 2^{-\alpha})H(p_2,p_1)
=2^{\alpha}\left (2^{-\alpha}H\left(\frac{1}{2},\frac{1}{2}\right)(1-p_1^{\alpha}-p_1^{\alpha})-2^{-2\alpha}H\left(p_1,p_2\right)\right ).
\end{align}
Moreover, one easily sees that \eqref{prefo} transforms to
\begin{align}\label{neufo2}
(1-3\cdot 2^{-\alpha})&H(p_2,p_1)\\
=2^{\alpha}&\left ((1-3\cdot 2^{-\alpha})H\left(\frac{1}{2},\frac{1}{2}\right)(1-p_1^{\alpha}-p_2^{\alpha})-(1-6\cdot 2^{-\alpha}+9\cdot 2^{-2\alpha})\, H(p_1,p_2)\right ).\nonumber
\end{align}
\eqref{neufo1} and \eqref{neufo2} provide
\begin{align*}
(1-2^{2-\alpha})H\left(\frac{1}{2},\frac{1}{2}\right)(1-p_1^{\alpha}-p_2^{\alpha})&=(1-3\cdot 2^{1-\alpha}+2^{3-2\alpha})\, H(p_1,p_2).
\end{align*}
Since $1-3\cdot 2^{1-\alpha}+2^{3-2\alpha}=(1-2^{2-\alpha})(1-2^{1-\alpha})$, it follows
\begin{align*}
H(p_1,p_2)=\frac{1-p_1^{\alpha}-p_2^{\alpha}}{1-2^{1-\alpha}}\ H\left(\frac{1}{2},\frac{1}{2}\right).
\end{align*}
In order to show (ii), let $\alpha=2$ and define maps
$f: [\frac{1}{2},1]\rightarrow [\frac{1}{2},1]$ and $D: [\frac{1}{2},1]\rightarrow [0,\infty[$ by
\begin{equation*}
f(p)=\max\left\{\frac{1-p}{p},1-\frac{1-p}{p}\right\}
\end{equation*}
and
\begin{equation*}
D(p)=|H(p,1-p)-H(1-p,p)|
\end{equation*}
for $p\in [\frac{1}{2},1]$.

By \eqref{prefo3} in Lemma \ref{prelem}, (a) is equivalent both to (b) and to $D(p)=0$ for all $p\in [\frac{1}{2},1]$.
(c) implies (d) by compactness of $\triangle_2$ and validity of the implications (a) $\Rightarrow$ (c) and (a) $\Rightarrow$ (e) is obvious.

From
\begin{eqnarray*}
H(p,1-p)+p^2H\left(\frac{1-p}{p},1-\frac{1-p}{p}\right)&=&H(1-p,2p-1,1-p)\\&=&H(1-p,p)+p^2H\left(1-\frac{1-p}{p},\frac{1-p}{p}\right)
\end{eqnarray*}
for $p\in\,\left [\frac{1}{2},1\right ]$ one obtains
\begin{equation}\label{dp}
D(p)=p^2 D(f(p))
\end{equation}
and by induction
\begin{equation}\label{prodl}
D(p)=\left (\prod_{k=1}^n f^{\circ k}(p)\right )^{\hspace{-1mm} 2} D(f^{\circ n}(p))
\end{equation}
with $f^{\circ n}(p)=\overbrace{f(f(\ldots (f(
}^{n\mbox{\footnotesize\ times}}p))\ldots )).$

For $p\in\,\left ]\frac{2}{3},1\right [$ it holds $f(p)=2-\frac{1}{p}$, hence $f$ maps the interval $\left ]\frac{2}{3},1\right [$ onto the interval $\left ]\frac{1}{2},1\right [$. Since $p-f(p)=\frac{(p-1)^2}{p}>0$
for all $p\in\,\left ]\frac{2}{3},1\right [$, the following holds:
\begin{eqnarray}\label{existsk}
\mbox{ For each }q\in\,\left]\frac{2}{3},1\right [\mbox{ there exists some }k\in {\mathbb N}\mbox{ with }f^{\circ k}(q)\in\,\left ]\frac{1}{2},\frac{2}{3}\right [.
\end{eqnarray}
Moreover, applying \eqref{dp} to $p=\frac{1}{2}$ yields $0=D(\frac{1}{2}) = \frac{D(1)}{4}$, hence
\begin{eqnarray}\label{d10}
D(1)=0.
\end{eqnarray}
Assuming (d), by use of \eqref{prodl}, \eqref{existsk} and \eqref{d10} one obtains
$D(p)=0$ for all $p\in\,\left [\frac{1}{2},1\right ]$, hence (a). If (e) is valid, then by \eqref{prefo3} in Lemma \ref{prelem}
\begin{eqnarray*}
D(r)\leq \left|4H\left (\frac{1}{2},\frac{1}{2}\right)\right | (1-r^2-(1-r)^2)\leq \left|4H\left (\frac{1}{2},\frac{1}{2}\right)\right |
\end{eqnarray*}
for all $r\in [\frac{1}{2},1]$, providing (d). By the already shown, (a), (b), (c), (d), (e) are equivalent .
\end{proof}	

We are able now to complete the proof of Theorem \ref{mainth}.
Assuming \eqref{qadd}, we first show \eqref{maini} for $\alpha\neq 1,2$, and for $H$ bounded and $\alpha=2$.
This provides statement (i) and, together with Lemma \ref{beginlem} (ii), statement (ii) of Theorem \ref{mainth}. 

Statement \eqref{maini} is valid for all $(p_1,p_2,\ldots ,p_n)\in\triangle_1\cup\triangle_2$ by Lemma \ref{beginlem}. In order to prove it for $n>2$, we use induction. Assuming validity of \eqref{maini} for all $(p_1,p_2,\ldots ,p_n)\in\triangle$ with $n=k$, where $k\in {\mathbb N}\setminus\{1\}$, let $(p_1,p_2,\ldots ,p_k,p_{k+1})\in\triangle$. Choose some
$j\in\{1,2,\ldots ,k\}$ with $p_j+p_{j+1}>0$. Then by \eqref{qadd} and Lemma \ref{beginlem} we have
\begin{align*}
H(p_1,\ldots,p_{j-1},p_j,p_{j+1},p_{j+2},\ldots ,p_{k+1})\hspace*{-5cm}&\\&=H(p_1,\ldots,p_{j-1},p_j+p_{j+1},p_{j+2},\ldots ,p_{k+1})
+(p_j+p_{j+1})^\alpha\, H\left(\frac{p_j}{p_j+p_{j+1}},\frac{p_{j+1}}{p_j+p_{j+1}}\right)\\
&=H\left(\frac{1}{2},\frac{1}{2}\right)\frac{1-\sum_{i=1}^{j-1} p_i^{\alpha}-(p_j+p_{j+1})^{\alpha }-\sum_{i=j+2}^{k+1} p_i^{\alpha}}{1-2^{1-\alpha}}\\
&\hspace{8mm} +H\left(\frac{1}{2},\frac{1}{2}\right)\frac{(p_j+p_{j+1})^\alpha}{1-2^{1-\alpha}}\left(1-\left (\frac{p_j}{p_j+p_{j+1}}\right )^{\hspace*{-1mm}\alpha}-\left (\frac{p_{j+1}}{p_j+p_{j+1}}\right )^{\hspace*{-1mm}\alpha}\right)\\
&=H\left(\frac{1}{2},\frac{1}{2}\right)\frac{1-\sum_{i=1}^{k+1} p_i^{\alpha}}{1-2^{1-\alpha}}.
\end{align*}
So \eqref{maini} holds for all $(p_1,p_2,\ldots ,p_n)\in\triangle$ with $n=k+1$.

In order to see (iii), recall a result of Diderrich \cite{Diderrich1975} stating that $H:\triangle\rightarrow{\mathbb R}$ is a multiple of the Shannon entropy if $H$ is bounded and permutation-invariant on $\triangle_2$ and satisfies
\begin{align*}
H(p_1,p_2,p_3,\ldots ,p_n)=H(p_1+p_2,p_3\ldots ,p_n)+(p_1+p_2)\, H\left(\frac{p_1}{p_1+p_{2}},\frac{p_2}{p_1+p_2}\right)\hspace{-0.5cm}\\
\mbox{for all }(p_1,p_2,\ldots, p_n)\in\triangle;\, n\in {\mathbb N}\setminus\{1\}\mbox{ with }p_1+p_2>0,\hspace*{2.77cm}
\end{align*}
which is weaker than \eqref{qadd} with $\alpha=1$. Since under \eqref{qadd} $H$ is permutation-invariant by Lemma \ref{permut}, Diderrichs axiom are satisfied, and we are done.

\section{Further discussion}
Our discussion suggests that the case $\alpha=2$ is more complicated than the general one.
In order to get some further insights, particularly in the case $\alpha=2$, let us consider consider only rational stochastic vectors.
So in the following let $\triangle^{\mathbb Q}=\bigcup_{n\in {\mathbb N}}\triangle_n^{\mathbb Q}$ with $\triangle_n^{\mathbb Q}=\triangle_n\cap {\mathbb Q}^n$ for $n\in {\mathbb N}$ and ${\mathbb Q}$ being the rationals. The following proposition states that for $\alpha\neq 1$ the `rational' generalized Shannon additivity principally provides the Tsallis entropy on the rationals, which particularly provides a proof of the implication (c) $\Rightarrow$ (a) in Theorem \ref{mainth} (ii).
\begin{pro}
Let $H: \triangle^{\mathbb Q}\rightarrow {\mathbb R}$ be given with \eqref{ks4} for $\triangle^{\mathbb Q}$ instead of $\triangle$ and $\alpha\in {\mathbb R}^+\setminus\{0,1\}$. Then it holds
\begin{align}\label{mainrat1}
H(p_1,p_2,\ldots ,p_n)=H\left(\frac{1}{2},\frac{1}{2}\right)\frac{1-\sum_{i=1}^n p_i^{\alpha}}{1-2^{1-\alpha}}\mbox{ for all }(p_1,p_2,\ldots ,p_n)\in\triangle^{\mathbb Q}.
\end{align}
\end{pro}
\begin{proof}
For the vectors $\left(\frac{1}{m}, \ldots, \frac{1}{m} \right), \left(\frac{1}{n}, \ldots, \frac{1}{n} \right)\in\triangle $ with $m,n \in {\mathbb N}$, we get from axiom \eqref{ks4}
\begin{align*}
H\left(\frac{1}{m},..., \frac{1}{m}\right)+m\left(\frac{1}{m}\right)^{\alpha}H\left(\frac{1}{n},..., \frac{1}{n}\right)&=H\left(\frac{1}{mn},..., \frac{1}{mn}\right)\\
&=H\left(\frac{1}{n},..., \frac{1}{n}\right)+n\left(\frac{1}{n}\right)^{\alpha}H\left(\frac{1}{m},..., \frac{1}{m}\right),
\end{align*}
implying
\begin{align}\label{proofRat1}
H\left(\frac{1}{m},..., \frac{1}{m}\right) = H\left(\frac{1}{n},..., \frac{1}{n}\right) \cdot \frac{1-\left(\frac{1}{m}\right)^{\alpha-1}}{1-\left(\frac{1}{n}\right)^{\alpha-1}}.
\end{align}
Now consider any rational vector $(p_1,p_2,\ldots ,p_n)\in \triangle^{\mathbb Q}$ with $p_1=\frac{a_1}{b}, p_2=\frac{a_2}{b},..., p_n=\frac{a_n}{b}$ for $b,a_1,\ldots ,a_n\in {\mathbb N}$ satisfying $\sum\limits_{i=1}^{n}a_i = b$. With \eqref{ks4} we get
\begin{align*}
&H(p_1,...,p_n) + \sum\limits_{i=1}^{n}p_i^\alpha\cdot H\left(\frac{1}{a_i\cdot n},..., \frac{1}{a_i\cdot n}\right) \\
&\quad\qquad = H\left(\frac{1}{b \cdot n},...,\frac{1}{b \cdot n}\right)
= H\left(\frac{1}{n},..., \frac{1}{ n}\right) + n\cdot \left(\frac{1}{n}\right)^\alpha \cdot H\left(\frac{1}{b},..., \frac{1}{b}\right).
\end{align*}
Using \eqref{proofRat1}, we obtain
\begin{align*}
H(p_1,\ldots,p_n)
&= H\left(\frac{1}{n},\ldots, \frac{1}{ n}\right)\! \cdot \! \left( 1
+ n\cdot\left(\frac{1}{n}\right)^\alpha \cdot\frac{1-\left(\frac{1}{b}\right)^{\alpha-1}}{1-\left(\frac{1}{n}\right)^{\alpha-1}}
- \sum\limits_{i=1}^{n} p_i^\alpha
\frac{1-\left(\frac{1}{a_i\cdot n}\right)^{\alpha-1}} {1-\left(\frac{1}{n}\right)^{\alpha-1}}
\right) \\
&= H\left(\frac{1}{n},\ldots, \frac{1}{ n}\right) \cdot
\frac{1 - \sum\limits_{i=1}^{n}p_i^\alpha}
{1-\left(\frac{1}{n}\right)^{\alpha-1}} = H\left(\frac{1}{2},\frac{1}{2}\right) \cdot
\frac{1 - \sum\limits_{i=1}^{n}p_i^\alpha}
{1-2^{1-\alpha}}.\qedhere
\end{align*}
\end{proof}
Let us finally compare (ii) and (iii) in Theorem \ref{mainth} and ask for the role of (c), (d) and (e) of (ii) in (iii). Symmetry is already given by Lemma \ref{permut} when only \eqref{ks4} is satisfied, \eqref{ks4} and nonnegativity are not sufficient for characterizing Shannon entropy, as shown in \cite{DaroczyMaksa79}. By our knowledge, there is no proof that \eqref{ks4} and continuity are enough, but \eqref{ks4} and analyticity is working. Showing the latter, in \cite{NambiarETAl1992} an argumentation reducing everything to the rationals as above has been used. 

We want to resume with the open problem whether the further assumptions for $\alpha=2$ in Theorem \ref{mainth} are necessary.
\begin{problem}
Is \eqref{maini} in Theorem \ref{mainth} also valid for $\alpha=2$?
\end{problem}


\begin{thebibliography}{99}

\bibitem{Abe2004} S.~Abe, Tsallis entropy: How unique?,
{\it Contin.~Mech.~Thermodyn.} 16 (2004), 237 -- 244.

\bibitem{Csiszar2008} I.~Csisz\'{a}r, Axiomatic characterizations of information measures, {\it Entropy} 10 (2008), 261 -- 273.

\bibitem{DaroczyMaksa79} Z.~Dar\'{o}czy, D.~Maksa, Nonnegative information functions, in: Analytic Function Methods in
Probability and Statistics, Colloq.~Math.~Soc.~J.~Bolyai 21, Gyires, B., Ed., North Holland,
Amsterdam 1979, 65 -- 76.

\bibitem{Diderrich1975} G.T.~Diderrich, The role of boundedness in characterizing Shannon entropy, {\it Inf.~and Control} 29 (1975), 140 -- 161.

\bibitem{Faddeev1956} D.K.~Faddeev,  On the concept of entropy of a finite probability scheme (in Russian), {\it Uspehi Mat.~Nauk} 1956, 227 –- 231.

\bibitem{Furuichi2005}
S.~Furuichi, On uniqueness Theorems for Tsallis entropy and Tsallis relative entropy,
{\it IEEE Trans.~Inf.~Theory} 51 (2005), 3638 -- 3645.

\bibitem{Cartwright2014}
J.~Cartwright, Roll over, Boltzmann, {\it Physics World} 27 (2014), 31 -- 35. 

\bibitem{HavrdaCharvat1967} J.~Havrda, F.~Charv\'{a}t,
Quantification method of classification processes. Concept of structural $\alpha$-entropy,
{\it Kybernetika} 3 (1967), 30 -- 35.

\bibitem{IlicEtAl2013} V.M.~Ili\'{c}, M.S.~Stankovi\'{c}, and E.H.~Mulali\'{c},
Comments on ``Generalization of Shannon-Khinchin axioms to nonextensive systems and the uniqueness theorem for the nonextensive entropy'',
{\it IEEE Trans.~Inf.~Theory} 59 (2013), 6950 -- 6952.

\bibitem{Khinchin1957} A.I.~Khinchin,
{\it Mathematical Foundations of Information Theory},
Dover, New York 1957.

\bibitem{NambiarETAl1992} K.K.~Nambiar, P.K.~Varma, and V.~Saroch, An axiomatic definition of Shannon's entropy, {\it App.~Math.~Lett.} 5 (1992), 45 -- 46.

\bibitem{dosSantos1997}
R.J.V.~dos Santos, Generalization of Shannon's theorem for Tsallis entropy,
{\it J.~Math.~Phys.} 38 (1997), 4104 -- 4107.

\bibitem{Shannon1948} C.E.~Shannon, A mathematical theory of communication,
{\it Bell Syst.~Tech.~J.} 27 (1948), 397 -- 423, 623 -- 653.

\bibitem{Suyari2004} H.~Suyari, Generalization of Shannon-Khinchin axioms to nonextensive systems and the uniqueness theorem for the nonextensive entropy,
{\it IEEE T.~Inform. Theory} 50 (2004), 1783 -- 1787.

\bibitem{Tsallis2016} C.~Tsallis, Approach of complexity in nature: Entropic nonuniqueness, {\it Axioms} 5 (2016), 5030020.

\bibitem{Tsallis1988} C.~Tsallis, Possible generalization of Boltzmann-Gibbs statistics, {\it J.~Stat.~Phys.} 52 (1988), 479 –- 487.


\end{thebibliography}
\end{document}